\documentclass{llncs}
\usepackage{latexsym}
\usepackage{amssymb}
\usepackage{amsmath}
\usepackage{stmaryrd}
\usepackage{array}
\usepackage{proof}
\outer\long\def\COUIC#1{}
\def\limp{\Rightarrow}
\def\liff{\Leftrightarrow}
\def\exelim{\delta_{\ex}}
\def\botelim{\delta_{\bot}}
\def\real{\Vdash}
\def\<{\langle}
\def\>{\rangle}

\def\ex{\exists}

\def\la{\leftarrow}

\def\T{\mathcal{T}}
\def\U{\mathcal{U}}
\def\L{\mathcal{L}}
\def\S{\mathcal{S}}

\def\Int#1{\llbracket#1\rrbracket}
\def\M{\mathcal{M}}
\def\C{\mathcal{C}}

\def\CR{\mathit{CR}}

\def\Elim{\mathit{Elim}}
\def\SN{\mathit{SN}}
\def\Code#1{\lceil#1\rceil}
\def\Edoc#1{\lfloor#1\rfloor}
\def\Term{\mathit{Term}}
\def\Nat{\mathit{Nat}}
\def\Le{\mathit{Le}}
\def\Sort{\mathit{Sort}}
\def\Proof{\mathit{Proof}}
\def\ProofVar{\mathit{ProofVar}}
\def\TermVar{\mathit{TermVar}}
\def\TSubst{\mathit{TSubst}}
\def\PSubst{\mathit{PSubst}}
\def\Red{\mathit{Red}}
\def\Redn{\mathit{Redn}}

\def\mathconstr#1{\mathit{#1}}
\def\Axiom{\mathconstr{Axiom}}
\def\BotE{\mathconstr{Bot\_E}}
\def\TopI{\mathconstr{Top\_I}}
\def\ImpI{\mathconstr{Imp\_I}}
\def\ImpE{\mathconstr{Imp\_E}}
\def\AndI{\mathconstr{And\_I}}
\def\AndEa{\mathconstr{And\_E}_1}
\def\AndEb{\mathconstr{And\_E}_2}
\def\OrIa{\mathconstr{Or\_I}_1}
\def\OrIb{\mathconstr{Or\_I}_2}
\def\OrE{\mathconstr{Or\_E}}
\def\ForallI{\mathconstr{Forall\_I}}
\def\ForallE{\mathconstr{Forall\_E}}
\def\ExistsI{\mathconstr{Exists\_I}}
\def\ExistsE{\mathconstr{Exists\_E}}

\def\fst{\mathit{fst}}
\def\snd{\mathit{snd}}

\begin{document}

\title{Relative normalization}
\author{Gilles Dowek\inst{1} \and Alexandre Miquel\inst{2}}
\institute{
  {\'E}cole polytechnique and INRIA\\
  LIX, {\'E}cole polytechnique, 91128 Palaiseau cedex, France \\
  \email{Gilles.Dowek@polytechnique.fr} \and
  Universit{\'e} Paris~7,\\
  PPS, 175 Rue du Chevaleret, 75013 Paris, France \\
  \email{Alexandre.Miquel@pps.jussieu.fr} }
\pagestyle{headings} 

\maketitle


G{\"o}del's second incompleteness theorem forbids to prove, in a given
theory~$\U$, the consistency of many theories---in particular, of the
theory~$\U$ itself---as well as it forbids to prove the normalization
property for these theories, since this property implies their
consistency.
When we cannot prove in a theory~$\U$ the consistency of a
theory~$\T$, we can try to prove a relative consistency theorem,
that is, a theorem of the form:
\begin{quote}
  \emph{If~$\U$ is consistent then~$\T$ is consistent}.
\end{quote}
Following the same spirit, we show in this paper how to prove relative
normalization theorems, that is, theorems of the form:
\begin{quote}
  \emph{If $\U$ is $1$-consistent, then~$\T$ has the
    normalization property.}
\end{quote}

\section{An abstract consistency result}
\label{s:AbstrConsistency}

To prove that a theory~$\T$ is consistent provided the theory~$\U$ is,
we usually assume given a model of~$\U$ and we build a model of~$\T$.
When the domain of the model of~$\T$ is a subset of the domain of the
model of~$\U$---which is called an internal model---we can however
proceed in a slightly different way:
instead of mapping a formula~$A$ of~$\T$ to a truth value
$\Int{A}_{\phi}$ (depending on an assignment~$\phi$)
we can translate it into a formula $A^*$ of~$\U$ that expresses that
the truth value associated to~$A$ is~$1$.
And instead of proving that $\T\vdash A$ entails $\Int{A}_{\phi}=1$,
we prove that $\T\vdash A$ entails $\U\vdash A^*$.
Finally, if~$\bot^*$ is $\bot$ (or any equivalent formula in~$\U$) and
if the theory~$\U$ is consistent, then the theory~$\T$ is consistent
too.

Unlike the method based on model extrusion, the method by which
formul{\ae} of~$\T$ are directly translated as formul{\ae} of~$\U$
does not require to take care of the free variables of~$A$ with an
assignment: free variables of~$A$ just remain free variables of~$A^*$,
after being possibly renamed.

This way of proving the consistency of~$\T$ is quite different from
proving in~$\U$ the existence of a model of~$\T$, as we do not need to
define a domain of interpretation: all the universe of discourse of
the theory~$\U$---or part of it---can serve as domain.
Thus this method can be applied also in cases where the
existence of a model of~$\T$ can not be proved in~$\U$.

\begin{definition}\label{d:Interp}
  --- An \emph{interpretation} of a theory~$\T$ in a theory~$\U$ is
  given by
  \begin{itemize}
  \item a function which maps each sort~$s$ of~$\T$ to a sort~$s_*$
    of~$\U$ with a relativization predicate~$s^*(x)$ in~$\U$ (with~$x$
    of sort~$s_*$), such that $\U\vdash\exists x~s^*(x)$.
  \item a function which maps each variable~$x$ of sort~$s$ in~$\T$ to
    a variable~$x^*$ of sort~$s_*$ in~$\U$;
  \item a function which maps each formula $A$ of~$\T$ with free
    variables $x_1,\ldots,x_n$ of sorts $s_1,\ldots,s_n$ in~$\T$
    to a formula $A^*$ of~$\U$ whose free variables occur among
    the variables $x_1^*,\ldots,x_n^*$
  \end{itemize}
  for which we require that
  \begin{enumerate}
  \item for all sorts~$s$ of~$\T$ we have
    $\U\vdash\exists x~s^*(x)$;
  \item for all formul{\ae}~$A$ and~$B$ of~$\T$ whose free
    variables occur among the variables $x_1,\ldots,x_n$ of sorts
    $s_1,\ldots,s_n$ in~$\T$ we have:
    $$\begin{array}{l}
      U\vdash \bot^*\liff\bot \\
      U\vdash \top^*\liff\top \\
      U\vdash s_1^*(x^*_1)\land\cdots\land s_n^*(x^*_n)~\limp~
      \bigl((A\land B)^*\liff(A^*\land B^*)\bigr) \\
      U\vdash s_1^*(x^*_1)\land\cdots\land s_n^*(x^*_n)~\limp~
      \bigl((A\lor B)^*\liff(A^*\lor B^*)\bigr) \\
      U\vdash s_1^*(x^*_1)\land\cdots\land s_n^*(x^*_n)~\limp~
      \bigl((A\limp B)^*\liff(A^*\limp B^*)\bigr) \\
      U\vdash s_1^*(x^*_1)\land\cdots\land s_{n-1}^*(x^*_{n-1})~\limp~
      \bigl((\forall x_n\,A)^*\liff
      \forall x^*_n\,(s_n^*(x^*_n)\limp A^*)\bigr) \\
      U\vdash s_1^*(x^*_1)\land\cdots\land s_{n-1}^*(x^*_{n-1})~\limp~
      \bigl((\exists x_n\,A)^*\liff
      \exists x^*_n\,(s_n^*(x^*_n)\land A^*)\bigr) \\
    \end{array}$$
  \item for all axioms~$A$ of~$\T$ we have:\quad $\U\vdash A^*$\quad
    (assuming that the axioms of~$\T$ are closed formul{\ae}).
  \end{enumerate}
\end{definition}

A simple way to define the underlying translation $A\mapsto A^*$ of
such an interpretation is to define it structurally on formul{\ae},
by first defining the formula $A^*$ for each atomic formula~$A$
of~$\T$, and then by extending the definition inductively to all
formul{\ae} using the equations
$$\begin{array}{r@{~~}c@{~~}l@{\qquad\qquad}r@{~~}c@{~~}l}
  (A\limp B)^* &\equiv& (A^*\limp B^*) &
  \multicolumn{3}{l}
  {\top^*~\equiv~\top\qquad\quad\bot^*~\equiv~\bot} \\
  (A\land B)^* &\equiv& (A^*\land B^*) &
  (\forall x\,A)^* &\equiv& \forall x^*\,(s^*(x^*)\limp A^*) \\
  (A\lor B)^* &\equiv& (A^*\lor B^*) &
  (\exists x\,A)^* &\equiv& \exists x^*\,(s^*(x^*)\land A^*) \\
\end{array}$$
(assuming that~$s$ is the sort of the variable~$x$).
In this case, the conditions required in item~1 of the definition
above come for free, but the soundness of all the axioms of~$\T$
in~$\U$ still has to be checked separately.

Notice that the definition of the notion of interpretation (of a
theory into another one) does not say anything about the translation
of terms.
The reason is that in some cases, it is desirable to interpret a
theory with a rich term language (such as Peano arithmetic)
into a theory
with a poor one (such as set theory), so that we can not expect that
terms of~$\T$ are always interpreted as terms of~$\U$.

However, interpretations are usually structural on terms too, which
means that they come with a translation $t\mapsto t^*$ on terms such
that
\begin{itemize}
\item for all terms $t$ of sort~$s$ whose free variables occur among
  the variables  $x_1,\ldots,x_n$ of sorts $s_1,\ldots,s_n$ in the
  theory~$\T$, one has:
  $$\U\vdash s^*_1(x^*_1)\land\cdots\land s^*_n(x^*_n)
  ~\limp~s^*(t^*)$$
\item for all well-formed terms of~$\T$ of the form
  $f(t_1,\ldots,t_n)$ (where~$f$ is an arbitrary function symbol
  of~$\T$), one has:\footnote{This condition says that each function
    symbol~$f$ of~$\T$ is translated as a macro in~$\U$.}
  $$(f(t_1,\ldots,t_n))^*=
  (f(x_1,\ldots,x_n))^*\{x^*_1:=t^*_1;\ldots;x^*_n:=t^*_n\}\,.$$
\end{itemize}

The main interest of the notion of interpretation (in the
sense of Def.~\ref{d:Interp}) is that it provides a way to translate
each theorem of~$\T$ into a theorem of~$\U$:
\begin{proposition}\label{interptheo}
  --- Given an interpretation $*$ of a theory~$\T$ in a theory~$\U$ and
  a formula~$A$ of~$\T$ whose free variables occur among the variables
  $x_1,\ldots,x_n$ of sorts $s_1,\ldots,s_n$ in~$\T$,
  if $\T\vdash A$ then 
  $$\U\vdash s_1^*(x^*_1) \land ... \land s_n^*(x^*_n) \limp A^*$$
\end{proposition}

\begin{proof}
  By induction on the structure of the proof of $\T\vdash A$.
\end{proof}

\begin{theorem}
  --- If the theory $\T$ has an interpretation in~$\U$, and if the
  theory~$\U$ is consistent, then~$\T$ is consistent too.
\end{theorem}

\begin{proof}
  Assume $\T\vdash\bot$.
  From Prop.~\ref{interptheo} we get $\U\vdash\bot^*$ and thus
  $\U\vdash\bot$ using the equivalence $\U\vdash\bot^*\liff\bot$.\qed
\end{proof}

Notice that this way of proving the consistency of a theory directly
extends to intuitionistic logic.

\section{Normalization}

\subsection{Deduction modulo}

We want to be able to speak abstractly of the normalization of proofs
in an arbitrary theory~$\T$ and also to be able to deduce some
corollaries from the fact that~$\T$ has the normalization property, in
particular, the consistency of~$\T$, the disjunction and the witness
property for constructive proofs, etc.
It is well-known that predicate logic is not a appropriate for
defining such a notion of normalization as each axiomatic theory~$\T$
requires a specific notion of reduction.
Thus we use an extension of  predicate logic called \emph{Deduction
  modulo} \cite{DHK}.

In Deduction modulo, a theory is formed is formed with a set of axioms
$\Gamma$ and a congruence $\equiv$ defined on formul{\ae}.
Then, the deduction rules take this congruence into account.
For instance, the \emph{modus ponens} is not stated as usual
$$\infer{B}{A\limp B & A}$$
as the first premise need not be exactly $A\limp B$ but may be
only congruent to this formula, hence it is stated
$$\infer[\text{if}~C\equiv A\limp B]{B}{C & A}$$
All the rules of natural deduction may be stated in a similar way.
Many theories, such as arithmetic, simple type theory and set theory
can be expressed with a congruence and no axioms.

Replacing axioms by a congruence changes the structure of
proofs and in particular some theories may have the normalization
property when expressed with axioms and not when expressed with
a congruence. 
In counterpart, when theory formed with a congruence and no axioms 
has the normalization
property, then we can deduce that it is consistent, 
constructive proofs have the disjunction and the witness property, 
various proof search methods are complete, etc.

\subsection{Pre-models}

A theory in Deduction modulo has the normalization property if it
has what we call a pre-model~\cite{DowekWerner}.
A \emph{pre-model} is a many-valued model whose truth values are
reducibility candidates, that is, particular sets of strongly
normalizable proof-terms whose definition is given below.

\begin{definition}[Proof-term]\label{def:ProofTerm}
  \emph{Proof-terms} are inductively defined as follows:
  $$\begin{array}{r@{~~}r@{~~}l@{\qquad}l}
    \pi &::=&\alpha & (\text{Axiom}) \\
    &|& \lambda\alpha~\pi \quad|\quad (\pi~\pi') &
    ({\limp}\text{-intro},~{\limp}\text{-elim})\\
    &|& \langle\pi,\pi'\rangle \quad|\quad
    \fst(\pi) \quad|\quad \snd(\pi) &
    ({\land}\text{-intro},~{\land}\text{-elim}_{1,2}) \\
    &|& i(\pi) \quad|\quad j(\pi) \quad|\quad
    (\delta~\pi_{1}~\alpha \pi_{2}~\beta\pi_{3}) &
    ({\lor}\text{-intro}_{1,2},~{\lor}\text{-elim}) \\
    &|& I & ({\top}\text{-intro}) \\
    &|& (\botelim~\pi) & ({\bot}\text{-elim}) \\
    &|& \lambda x~\pi \quad|\quad (\pi~t) &
    ({\forall}\text{-intro},~{\forall}\text{-elim}) \\
    &|& \langle t,\pi\rangle \quad|\quad
    (\exelim~\pi~x \alpha \pi') &
    ({\exists}\text{-intro},~{\exists}\text{-elim}) \\
  \end{array}$$
\end{definition}
Each proof-term constructor corresponds to an inference rule of
intuitionistic natural deduction (see above).
A proof-term built using a constructor that corresponds to an
introduction rule---that is, a proof-term of the form
$\lambda\alpha~\pi$, $\langle\pi,\pi'\rangle$, $i(\pi)$, $j(\pi)$,
$I$, $\lambda x~\pi$ or $\langle t,\pi\rangle$---is called an
\emph{introduction}.
Similarly, a proof-term built using a constructor that corresponds to
an elimination rule---that is, a proof-term of the form
$(\pi~\pi')$, $\fst(\pi)$, $\snd(\pi)$,
$(\delta~\pi_1~\alpha\pi_2~\beta\pi_3)$,
$(\botelim~\pi)$, $(\pi~t)$ or $(\exelim~\pi~x\alpha\pi')$---is called
an \emph{elimination}.

\begin{definition}[Reduction]
  \emph{Reduction} on proof-terms is defined by the following rules
  that eliminate cuts step by step.
  $$\begin{array}{r@{\quad}c@{\quad}l}
    (\lambda \alpha~\pi_{1}~\pi_{2}) &\triangleright&
    \pi_{1}(\alpha\la\pi_{2}) \\
    \fst(\langle\pi_1,\pi_2\rangle) &\triangleright& \pi_1 \\
    \snd(\langle\pi_1,\pi_2\rangle) &\triangleright& \pi_2 \\
    (\delta~i(\pi_{1})~\alpha \pi_{2}~\beta \pi_{3}) 
    &\triangleright& \pi_{2}(\alpha \la \pi_1) \\
    (\delta~j(\pi_{1})~\alpha \pi_{2}~\beta \pi_{3}) 
    &\triangleright& \pi_{3}(\beta \la \pi_1) \\
    (\lambda x~\pi~t) &\triangleright& \pi(x \la t) \\
    (\exelim~\langle t,\pi_{1} \rangle~\alpha x\pi_{2}) &\triangleright&
    \pi_{2}(x \la t,\alpha \la \pi_1) \\
  \end{array}$$
\end{definition}

\begin{definition}[Reducibility candidates]
  A set $R$ of proof-terms is a \emph{reducibility candidate} if:
  \begin{itemize}
  \item if $\pi \in R$, then $\pi$ is strongly normalizable;
  \item if $\alpha$ is a variable, then $\alpha\in R$;
  \item if $\pi\in R$ and $\pi\triangleright\pi'$, then $\pi'\in R$;
  \item if $\pi$ is an elimination
    and if $\pi' \in R$ for all $\pi'$ such that
    $\pi\triangleright^{1}\pi'$,
    then  $\pi \in R$. 
  \end{itemize}
\end{definition}
The set of reducibility candidates is written~$\C$.

\begin{definition}[Pre-model]
  A \emph{pre-model} $\M$ (of a given signature) is given by
  \begin{itemize}
  \item A nonempty set (still) written~$\M$ and called the \emph{domain}
    of~$\M$ (or several domains $\M_{\sigma}$ for many-sorted
    theories);
  \item For each function symbol~$f$ of arity~$n$, a function
    $f^{\M}:\M^n\to\M$;
  \item For each predicate symbol~$p$ of arity~$n$, a function
    $p^{\M}:\M^n\to\C$.
  \end{itemize}
\end{definition}

If $P$ is an atomic formula and $\phi$ is a assignment on a
premodel~$\M$, the reducibility candidate $\Int{P}_{\phi}$ is defined
in the obvious way.
This definition extends to all formul{\ae} as follows:

\begin{itemize}
\item A proof-term is an element of $\Int{A\limp B}_{\phi}$ if it
  is strongly normalizable and when it reduces to a proof-term of the
  form $\lambda\alpha~\pi_1$, then for every $\pi'$ in
  $\Int{A}_{\phi}$, $[\pi'/\alpha]\pi_1$ is an element of
  $\Int{B}_{\phi}$.
\item A proof-term is an element of $\Int{A\land B}_{\phi}$ if it is
  strongly normalizable and when it reduces to a proof-term of the
  form $\langle\pi_1,\pi_2\rangle$, then~$\pi_1$ and~$\pi_2$ are
  elements of $\Int{A}_{\phi}$ and $\Int{B}_{\phi}$.
\item A proof-term is an element of $\Int{A\lor B}_{\phi}$ if it is
  strongly normalizable and when it reduces to a proof-term of the
  form $i(\pi_1)$ (resp.\ $j(\pi_2)$) then $\pi_1$ (resp.\ $\pi_2$)
  is an element of $\Int{A}_{\phi}$ (resp. $\Int{B}_{\phi}$). 
\item A proof-term is an element of $\Int{\top}_{\phi}$ if it is
  strongly normalizable.
\item A proof-term is an element of $\Int{\bot}_{\phi}$ if it is
  strongly normalizable.
\item A proof-term is an element of $\Int{\forall x~A}_{\phi}$ if
  it is strongly normalizable and when it reduces to a proof-term of
  the form $\lambda x~\pi_1$ then for every term~$t$ and every
  element~$v$ of $\M$,
  $[t/x]\pi_1$ is an element of
  $\Int{A}_{\phi+\langle x,v\rangle}$,
  where $\phi+\langle x,v\rangle$ is the function
  that coincides with $\phi$ everywhere except on $x$ where it takes
  value~$v$.
\item A proof-term is an element of $\Int{\exists x~A}_{\phi}$ if it
  is strongly normalizable and when it reduces to a proof-term of the
  form $\langle t,\pi_1\rangle$ there exists an element~$v$ of~$\M$
  such that $\pi_1$ is an element of
  $\Int{A}_{\phi+\langle x,v\rangle}$.
\end{itemize}

A pre-model~$\M$ is a pre-model of a congruence $\equiv$ if
for all formul{\ae}~$A$ and~$B$ such that $A \equiv B$, we have
$\Int{A}_{\phi}=\Int{B}_{\phi}$ for all assignments~$\phi$.

The theorem that if a theory has a pre-model then it has the
normalization property is proved in two steps.
We first check that
(1) for each formula~$A$ and assignment~$\phi$, the set
$\Int{A}_{\phi}$ is a reducibility candidate and then that
(2) proofs of $A$ are all members of the set $\Int{A}_{\phi}$.
The proof of (1) is an induction over the structure of formul{\ae} and
the proof of (2) an induction over the structure of proofs.
We need to prove, for each deduction rule a lemma such as:
\begin{quote}\it
  If $\pi_1\in\Int{A\limp B}_{\phi}$ and $\pi_2\in\Int{A}_{\phi}$,
  then $(\pi_1~\pi_2)\in\Int{B}_{\phi}$.
\end{quote}
Then, as all elements of a reducibility candidate strongly normalize
we conclude that all proofs of a formula~$A$ strongly
normalize~\cite{DowekWerner}.

\subsection{A theory of syntactic constructions}
\label{ss:SyntConstr}

To relativize the pre-model construction of~$\T$ w.r.t.\ a
theory~$\U$, we need to express all the syntactic constructions of the
proof-language of~$\T$ in~$\U$.
The proof-language of~$\T$ is complex: it contains proof-variables,
proof-terms, as well as the terms of the theory~$\T$ (that appear in
proof-terms).
Moreover, we need to express usual syntactic operations, such
as $\alpha$-conversion, substitution, etc.

For that let us consider a \emph{language of trees}~$\L$ generated by
a finite number of constructors, that is,
an algebra of closed terms generated by a finite number of function
symbols written $c_1,\ldots,c_N$---the \emph{constructors}
of~$\L$---whose arities are written $k_1,\ldots,k_N$.
In what follows, we assume that
the language~$\L$ provides two constructors~$0$ (of arity~$0$) and~$s$
(of arity~$1$) to encode natural numbers.

It is well-known that the latter assumption is sufficient to ensure
that all syntactic constructions (such as variables, terms,
proof-terms, etc.) can be encoded in~$\L$, by the mean of G{\"o}del
numberings.
However, the cost of these numberings can be avoided by taking
a richer language~$\L$, in which syntactic constructions can be
encoding more directly.

Once the language $\L$ has been fixed, the class of all functions that
can be defined by primitive recursion on~$\L$ is well-defined too.

From the language~$\L$, we build a (mono-sorted) first-order
theory~$\S$, which we call the
\emph{theory of syntactic constructions}.
This theory is defined as follows:
\begin{enumerate}
\item
  The function symbols of~$\S$ are the constructors $c_1,\ldots,c_N$
  plus, for each primitive recursive definition of a $n$ary function
  on syntactic trees, a function symbol $f$ of arity~$n$.
\item
  The only predicate symbol of the theory~$\S$ is equality.
\item
  The axioms of~$\S$ are:
  \begin{itemize}
  \item The equality axioms, that is:
    reflexivity, symmetry and transitivity of
    equality, as well as congruence axioms for all function
    symbols~$f$:
    $$\begin{array}{l}
      x=x \\
      x=y\land x=z\limp y=z \\
      x_1=y_1\land\cdots\land x_n=y_n
      ~~\limp~~ f(x_1,\ldots,x_n)=f(y_1,\ldots,y_n) \\
    \end{array}$$
  \item Axioms expressing injectivity and non-confusion for
    constructors:
    $$\begin{array}{r@{~~}c@{~~}l}
      c_i(x_1,\ldots,x_{k_i})=c_i(y_1,\ldots,y_{k_i})
      &\limp& x_1=y_1\land\cdots\land x_{k_i}=y_{k_i} \\
      c_i(x_1,\ldots,x_{k_i})=c_j(y_1,\ldots,y_{k_j})
      &\limp& \bot
    \end{array}$$
    (for all $i,j\in[1..N]$ such that $i\neq j$).
  \item For each primitive recursive definition of a function
    represented by a function symbol~$f$,
    the axioms expressing its equational theory.
  \item For each formula $A(x)$ possibly depending on a variable~$x$
    (as well as other parameters that are left implicit), the
    induction principle:
    $$\bigwedge_{i=1}^N
    \Bigl(\forall x_1~\cdots~\forall x_{k_i}~
    A(c_i(x_1,\ldots,x_{k_i}))\Bigr)
    ~~\limp~~\forall x~A(x)\,.$$
  \end{itemize}
\end{enumerate}


In what follows, we assume that the usual syntactic structures of~$\T$
are represented in the theory~$\S$ using the following predicates:
$$\begin{tabular}{l@{\qquad}l}
  $\Nat(x)$ & $x$ is a natural number \\
  $\Le(x,y)$ & $x$ is less than or equal to~$y$ \\
  $\Sort(x)$ & $x$ is a sort \\
  $\TermVar(x,y)$ & $x$ is a term variable of sort~$y$ \\
  $\Term(x,y)$ & $x$ is a term of sort~$y$ \\
  $\ProofVar(x)$ & $x$ is a proof variable \\
  $\Proof(x)$ & $x$ is a proof \\
  $\Elim(x)$ & $x$ is a proof that ends with an elimination rule \\
  $\Red(x,y)$ & the proof $x$ reduces in one step to the proof $y$ \\
  $\Redn(x,n,y)$ & the proof $x$ reduces in $n$ steps to the
  proof $y$ \\
\end{tabular}$$
Note that all the relations above are primitive recursive.

From these relations, we can define important (and non primitive
recursive) relations such as:
$$\begin{array}{r@{~~}c@{~~}l}
  \Red^*(x,y) &\equiv& \exists n~(\Nat(n)\land\Redn(x,n,y)) \\[6pt]
  \SN(x) &\equiv& \Proof(x)~\land\\
  && \exists n~\bigl(\Nat(n)~\land~
  \forall y~(\Proof(y)\limp\lnot\Redn(x,n,y))\bigr) \\
\end{array}$$
(Intuition: a proof~$\pi$ is strongly normalizable if there is a
number $n\ge 0$ such that $\pi$ has no $n$-reduct.)

We assume that proof-terms (Def.~\ref{def:ProofTerm}) are represented
as trees of~$\L$ by the mean of constructors---or constructor
aggregates---$\Axiom$ (arity~1),
$\ImpI$ (arity~2), $\ImpE$ (arity~2),
$\AndI$ (arity~2), $\AndEa$ (arity~1), $\AndEb$ (arity~1),
$\OrIa$ (arity~1), $\OrIb$ (arity~1), $\OrE$ (arity~5),
$\TopI$ (arity~0), $\BotE$ (arity~1),
$\ForallI$ (arity~2), $\ForallE$ (arity~2),
$\ExistsI$ (arity~2), $\ExistsE$ (arity~2)
whose name are self-explanatory
(see Def.~\ref{def:ProofTerm} for the correspondence).

Finally, we write $\TSubst$ (resp.\ $\PSubst$) the ternary function
symbol of~$\S$ that computes term-substitution
(resp.\ proof-substitution) inside a proof.

\section{An abstract normalization result}

We now want to relativize the pre-model construction. Hence, besides
the theory~$\T$ we want to prove the normalization of, we shall
consider another theory~$\U$ either in predicate logic or in deduction
modulo. 

\begin{definition} --- We say that a theory~$\U$ \emph{expresses
  syntactic constructions} if it comes with an interpretation
  (in the sense of Def.~\ref{d:Interp}) of the theory~$\S$ (defined
  in~\ref{ss:SyntConstr}) into~$\U$, which is structural
  on terms and formul{\ae} (cf section~\ref{s:AbstrConsistency}).
\end{definition}

From now on, we assume that~$\U$ expresses all syntactic
constructions, by the mean of an interpretation of the theory~$\S$
that we write using `$\Code{~}$' brackets (`the code of \dots').
In particular, we write $\Code{t}$ the code of any tree~$t\in\L$,
including (representations of) variables, terms and proof-terms of
the theory~$\T$.
The sort of~$\U$ associated to the unique sort of~$\S$ through the
interpretation of~$\S$ in~$\U$ is written $\Edoc{\L}$, and the
corresponding relativisation predicate is written $\Code{\L}(x)$.

Each primitive recursive relation~$R$ of arity~$n$ on~$\L$ is
expressed in~$\S$ as a relation still written
$R(x_1,\ldots,x_n)$ and defined by
$$R(x_1,\ldots,x_n)\quad\equiv\quad
f(x_1,\ldots,x_n)=1\,,$$
where $f$ is the function symbol of~$\S$ associated to the
characteristic function of~$R$.
Via the interpretation~$\Code{~}$, the primitive recursive
relation~$R$ is thus represented in~$\U$ as a relation written
$\Code{R}(x_1,\ldots,x_n)$.
It is important to notice that this representation is faithful when
the theory~$\U$ is consistent:
\begin{proposition} --- If $\U$ is consistent, then
  for all primitive recursive relations $R(t_1,\ldots,t_n)$ on~$\L$
  and for all $t_1,\ldots,t_n\in\L$ we have:
  \begin{center}
    $R(t_1,\ldots,t_n)$\qquad iff\qquad
    $\U\vdash\Code{R}(\Code{t_1},\ldots,\Code{t_n})$
  \end{center}
  In particular, we have:\quad
  $t_1=t_2$\quad iff\quad $\U\vdash\Code{t_1=t_2}$\quad
  (for all $t_1,t_2\in\L$).
\end{proposition}

\begin{proof}
  (\textit{Direct implication})\quad
  Assume that $R(t_1,\ldots,t_n)$ holds.
  From the trace of the computation of the characteristic
  function~$f$ of~$R$ applied to $t_1,\ldots,t_n$ we easily build
  a proof of $\S\vdash f(t_1,\ldots,t_n)=1$, from which we immediately
  get a proof of $\U\vdash\Code{R}(\Code{t_1},\ldots,\Code{t_n})$
  through the interpretation of~$\S$ in~$\U$.\smallbreak\noindent
  (\textit{Converse implication})\quad
  Assume that $\U\vdash\Code{R}(\Code{t_1},\ldots,\Code{t_n})$.
  We distinguish two cases,
  depending on whether $R(t_1,\ldots,t_n)$ holds or not.
  \begin{itemize}
  \item Either $R(t_1,\ldots,t_n)$ holds.\quad
    In this case we are done.
  \item Either $R(t_1,\ldots,t_n)$ does not hold.\quad
    From the trace of the computation of the characteristic
    function of~$R$, we now get a proof of
    $\S\vdash f(t_1,\ldots,t_n)=0$, and thus a proof of
    $\U\vdash\lnot\Code{R}(\Code{t_1},\ldots,\Code{t_n})$,
    which is impossible since $\U$ is consistent.
    Hence this case is absurd.\qed
  \end{itemize}
\end{proof}

In what follows, we will need a stronger notion of consistency,
namely:
\begin{definition}[$1$-consistency]\label{d:1Cons}
  --- We say that~$\U$ is \emph{$1$-consistent} if for all
  primitive recursive relations $R(t_1,\ldots,t_n)$ on~$\L$,
  the derivability of 
  $$\U\vdash\exists x_1\cdots\exists x_n~
  \bigl(\Code{\L}(x_1)\land\cdots\land\Code{\L}(x_n)\land
  \Code{R}(x_1,\ldots,x_n)\bigr)$$
  entails the existence of trees $t_1,\ldots,t_n\in\L$ such that
  $R(t_1,\ldots,t_n)$.
\end{definition}

The $1$-consistency of a theory entails its consistency, but the
converse does not hold in general.

\begin{definition}[Realizability translation]
  --- A \emph{realizability translation} of a theory~$\T$ in~$\U$ is
  defined by
  \begin{itemize}
  \item a function which maps each sort~$s$ of~$\T$ to a sort~$s_*$
    of~$\U$ with a relativization predicate~$s^*(x)$ in~$\U$ (with~$x$
    of sort~$s_*$);
  \item a function which maps each variable~$x$ of sort~$s$ in~$\T$ to
    a variable~$x^*$ of sort~$s_*$ in~$\U$;
  \item a function which maps each function symbol $f$ of rank
    $\<s_1,\ldots,s_n,s\>$ in~$\T$ to a term $f^*(z_1,\ldots,z_n)$
    of sort $s_*$ in~$\U$, that (possibly) depends on variables
    $z_1,\ldots,z_n$ of sort $s_{1*},\ldots, s_{n*}$, respectively.
  \item a function which maps each predicate symbol $p$ of rank
    $\<s_1,\ldots,s_n\>$~in $\T$ to a formula of~$\U$ written
    $\pi\real p(z_1,\ldots,z_n)$, that possibly depends on variables
    $\pi,z_1,\ldots,z_n$ of sorts $\Edoc{\L},s_{1*},\ldots,s_{n*}$,
    respectively.
  \end{itemize}
\end{definition}

Given a realizability translation of~$\T$ in~$\U$, the translation of
variables and function symbols is extended to all terms by setting:
$$\begin{array}{l@{~~}c@{~~}l}
  (x)^* &=& x^* \\
  (f(t_1,\ldots,t_n))^* &\equiv&
  f^*(z_1,\ldots,z_n)(z_1\la t^*_1,\ldots,z_n\la t^*_n)
\end{array}$$
Notice that this translation transforms each term~$t$ of~$\T$ with
free variables $x_1,\ldots,x_n$ of sorts $s_1,\ldots,s_n$ into a term
$t^*$ of~$\U$ whose free variables occur among the variables
$x^*_1,\ldots,x^*_n$ of sorts $s_{*1},\ldots,s_{*n}$.

Similarly, we extend the translation of predicate symbols to all
formul{\ae} by associating to
each formula~$A$ of~$\T$ with free variables $x_1,\ldots,x_n$ of sorts
$s_1,\ldots,s_n$
a formula $\pi\real A$ of~$\U$ with free variables
$\pi,x^*_1,\ldots,x^*_n$ of sorts $\Edoc{\L},s_{*1},\ldots,s_{*n}$.
The formula $\pi\real A$ is recursively defined on~$A$ by the
equations:
$$\begin{array}{l@{~}c@{~}l}
  \pi\real p(t_1,\ldots,t_n) &\equiv&
  (\pi\real p(z_1,\ldots,z_n))(z_1\la t^*_1,\ldots,z_n\la t^*_n)
  \\[6pt]
  \pi\real\top &\equiv& \Code{\SN}(\pi) \\[6pt]
  \pi\real\bot &\equiv& \Code{\SN}(\pi) \\[6pt]
  \pi\real A\limp B &\equiv&
  \Code{\SN}(\pi)~~\land\\
  && \forall\alpha~\forall\pi'~(
  \Code{\Red^*}(\pi,\Code{\ImpI}(\alpha,\pi')) ~\limp \\
  &&\hphantom{\forall\alpha~\forall \pi'~(}
  \forall\phi~(\phi\real A ~\limp~
  \Code{\PSubst}(\pi',\alpha,\phi)\real B)) \\[6pt]
  \pi\real A\land B &\equiv&
  \Code{\SN}(\pi)~~\land\\
  && \forall\pi_1\,\forall\pi_2~(
  \Code{\Red^*}(p,\Code{\AndI}(\pi_1,\pi_2)) ~\limp~
  \pi_1\real A~{\land}~\pi_2\real B) \\[6pt]
  \pi\real A\lor B &\equiv&
  \Code{\SN}(\pi)~~\land\\
  &&\forall\pi_1~(\Code{\Red^*}(\pi,\Code{\OrIa}(\pi_1))
  ~\limp~\pi_1\real A)~\land\\
  &&\forall\pi_2~(\Code{\Red^*}(\pi,\Code{\OrIb}(\pi_2))
  ~\limp~\pi_2\real B) \\[6pt]
  \pi\real\forall x~A(x) &\equiv&
  \Code{\SN}(\pi)~~\land\\
  && \forall v~\forall\pi'~(
  \Code{\Red^*}(\pi,\Code{\ForallI}(v,\pi')) ~\limp\\
  &&\hphantom{\forall v~\forall\pi'~(}
  \forall x^*~\forall t~(s^*(x^*)\land\Code{\Term}(t,\Code{s})~\limp\\
  &&\hphantom{\forall v~\forall\pi'~(\forall x^*~\forall t~(}
  \Code{\TSubst}(\pi',v,t)\real A(x))) \\[6pt]
  \pi\real\exists x~A(x) &\equiv&
  \Code{\SN}(\pi)~~\land\\
  && \forall\pi'~\forall t~(
  \Code{\Red^*}(\pi,\Code{\ExistsI}(t,\pi')) ~\limp\\
  &&\hphantom{\forall\pi'~\forall t~(}
  \exists x^*~(s^*(x^*)\land\pi'\real A(x))) \\[6pt]
\end{array}$$

We now need to express in the theory~$\U$ that the `set' (in its
informal sense) of all proofs~$\pi$ such that $\pi\Vdash A$ forms a
reducibility candidate.
Given a formula of~$\U$ possibly depending on a variable~$\pi$ of
sort~$\Edoc{\L}$ we write
$$\begin{array}{r@{\quad}c@{\quad}l@{\quad}l}
  \CR_{\pi}(A(\pi)) &\equiv&
  \forall\pi~\bigl(A(\pi)~\limp~
  \Code{\Proof}(\pi)\land\Code{\SN}(\pi)\bigr) &\land \\
  && \forall\alpha~\bigl(\Code{\ProofVar}(\alpha)~\limp~
  A(\Code{\Axiom}(\alpha))\bigr) &\land\\
  && \forall\pi~\forall\pi'~\bigl(A(\pi)\land\Code{\Red}(\pi,\pi')
  ~\limp~A(\pi')\bigr) &\land \\
  && \forall\pi~\bigl(\Code{\Elim}(\pi)\land
  \forall\pi'(\Code{\Red}(\pi,\pi')\limp A(\pi'))
  ~\limp~A(\pi)\bigr) \\
\end{array}$$

We can now introduce the conditions which make that a realizability
translation is a realizability interpretation:
\begin{definition}[Realizability interpretation]
  --- A realizability translation of~$\T$ in~$\U$ is a
  \emph{realizability interpretation} if the following conditions
  hold:
  \begin{enumerate}
  \item For each sort~$s$ of~$\T$, we have
    $\U\vdash\exists x~s^*(x)$;
  \item For each function symbol $f$ of rank $\<s_1,\ldots,s_n,s\>$
    in~$\T$:
    $$\U\vdash\forall z_1\cdots\forall z_n~
    \bigl(s^*_1(z_1)\land\cdots\land s^*_n(z_n)~\limp~
    s^*(f^*(z_1,\ldots,z_n))\bigr)$$
  \item For each predicate symbol~$p$ of rank $\<s_1,\ldots,s_n\>$
    in~$\T$:
    $$\U\vdash\forall z_1\cdots\forall z_n~
    \bigl(s^*_1(z_1)\land\cdots\land s^*_n(z_n)~\limp~
    \CR_{\pi}(\pi\real p(z_1,\ldots,z_n))\bigr)$$
  \item For all pairs of congruent formul{\ae} $A\equiv A'$ with free
    variables $x_1,\ldots,x_n$ of sorts
    $s_1,\ldots,s_n$:
    $$\U\vdash\forall x^*_1\cdots\forall x^*_n~
    \bigl(s^*_1(x^*_1)\land\cdots\land s^*_n(x^*_n)~\limp~
    \forall\pi(\pi\real A~\liff~\pi\real A')\bigr)$$
\COUIC{
  \item For all axioms~$A$ of~$\T$, there is a proof term $\pi$ such
    that $\U\vdash(\Code{\pi}\real A)$.
}
  \end{enumerate}
\end{definition}

Items~2 and~3 immediately extend to all terms and formul{\ae} as
follows:
\begin{proposition}[Typing]
  --- Given a realizability translation of~$\T$ in~$\U$:
  \begin{itemize}
  \item For all terms~$t$ of sort~$s$ in~$\T$ with free variables
    $x_1,\ldots,x_n$ of sorts $s_1,\ldots,s_n$:
    $$\U\vdash\forall x^*_1\cdots\forall x^*_n~
    \bigl(s^*_1(x^*_1)\land\cdots\land s^*_n(x^*_n)~\limp~
    s^*(t)\bigr)$$
  \item For all formul{\ae}~$A$ of~$\T$ with free variables
    $x_1,\ldots,x_n$ of sorts $s_1,\ldots,s_n$:
    $$\U\vdash\forall x^*_1\cdots\forall x^*_n~
    \bigl(s^*_1(x^*_1)\land\cdots\land s^*_n(x^*_n)~\limp~
    \CR_{\pi}(\pi\real A)\bigr)$$
  \end{itemize}
\end{proposition}

\begin{corollary}[Normalization of realizers]\label{cor:NormReal}
  --- For all formul{\ae}~$A$ of~$\T$ with free variables
  $x_1,\ldots,x_n$ of sorts $s_1,\ldots,s_n$:
  $$\U\vdash\forall x^*_1\cdots\forall x^*_n~\forall\pi~
  \bigl(s^*_1(x^*_1)\land\cdots\land s^*_n(x^*_n)\land
  \pi\real A~\limp~\Code{\SN}(\pi)\bigr)$$
\end{corollary}

We now extend the realizability relation to sequents as follows:
given a sequent $A_1,\ldots,A_k\vdash B$ of~$\T$ and a variable~$\pi$
of sort $\Edoc{\L}$ in~$\U$, we write
$$\begin{array}{l}
  \pi\real(A_1,\ldots,A_k\vdash B)~~\equiv\\[3pt]
  \quad\Code{\ImpI}(\Code{\alpha_1},\cdots
  \Code{\ImpI}(\Code{\alpha_k},\pi)\cdots)\real
  A_1\limp\cdots\limp A_k\limp B\,, \\
\end{array}$$
where $\alpha_1,\ldots,\alpha_k$ are pairwise distinct proof-variables
(of~$\T$) that represent the assumptions $A_1,\ldots,A_k$.
Corollary~\ref{cor:NormReal} immediately extends to sequent realizers:
\begin{corollary}\label{cor:NormRealSeq}
  --- For all sequents $\Gamma\vdash A$ of~$\T$ with free variables
  $x_1,\ldots,x_n$ of sorts $s_1,\ldots,s_n$:
  $$\U\vdash\forall x^*_1\cdots\forall x^*_n~\forall\pi~
  \bigl(s^*_1(x^*_1)\land\cdots\land s^*_n(x^*_n)\land
  \pi\real(\Gamma\vdash A)~\limp~\Code{\SN}(\pi)\bigr)$$
\end{corollary}

\begin{proposition}[Existence of a realizer]\label{prop:ExistReal}
  --- Let $(\_)^*$ be a realizability interpretation of~$\T$ in~$\U$.
  If a sequent $A_1,\ldots,A_k\vdash B$ of~$\T$ with free variables
  $x_1,\ldots,x_n$ of sorts $s_1,\ldots,s_n$
  has a proof~$\pi$ in
  intuitionistic deduction modulo~$\T$ with free proof-variables
  $\alpha_1:A_1$, \dots, $\alpha_k:A_k$, then
  $$\begin{array}{rcl}
    \U &\vdash& \forall x^*_1\cdots\forall x^*_n~
    \bigl(s^*_1(x^*_1)\land\cdots\land s^*_k(x^*_n)~\limp~
    \Code{\pi}\real(A_1,\ldots,A_k\vdash B)\bigr)
  \end{array}$$
\end{proposition}

\begin{proof}
  By induction on the derivation of $\pi:(A_1,\ldots,A_n\vdash B)$.
\end{proof}

\begin{theorem} --- If a theory $\T$ has a realizability
  interpretation in~$\U$ and if $\U$ is $1$-consistent,
  then $\T$ enjoys the strong normalization property.
\end{theorem}

\begin{proof}
  Assume $\pi$ is a proof of $A_1,\ldots,A_n\vdash B$ (in
  intuitionistic deduction modulo~$\T$).
  From Prop.~\ref{prop:ExistReal} we get
  $$\U\vdash\forall x^*_1\cdots\forall x^*_k~
  \bigl(s^*_1(x^*_1)\land\cdots\land s^*_k(x^*_k)~\limp~
  \Code{\pi}\real(A_1,\ldots,A_n\vdash B)\bigr)\,,$$
  hence
  $$\U\vdash\forall x^*_1\cdots\forall x^*_k~
  \bigl(s^*_1(x^*_1)\land\cdots\land s^*_k(x^*_k)\limp
  \Code{\SN}(\Code{\pi})\bigr)$$
  using Cor.~\ref{cor:NormRealSeq}, and finally
  $$\U\vdash\Code{SN}(\Code{\pi})\,,$$
  from the fact that all domains of interpretation of sorts are
  inhabited.  The latter means that
  $$\U\vdash\Code{\Proof}(\Code{\pi})~\land~
  \exists n~(\Code{\Nat}(n)\land\Code{R}(x,n))\,,$$
  writing
  $$R(x,n)~~\equiv~~
  \forall y~(\Proof(y)\limp\lnot\Redn(x,n,y))$$
  the relation of~$\S$ which expresses that $x$ has no $n$-reduct.
  But since this relation is primitive recursive, and since~$\U$ is
  $1$-consistent, we deduce that there exists a natural
  number~$n$ such that~$\pi$ has no $n$-reduct.
  Which means that~$\pi$ is strongly normalizable.\qed
\end{proof}

\bibliographystyle{plain}
\bibliography{realizmod}

\end{document}